\documentclass[sigconf]{aamas} 
\pdfoutput=1
\usepackage[british]{babel}

\setcopyright{ifaamas}
\acmConference[AAMAS '21]{
Proc.\@ of the 20th International Conference on
Autonomous Agents and Multiagent Systems (AAMAS 2021)}{May 3--7, 2021}{Online}{U.~Endriss, A.~Now\'{e}, F.~Dignum, A.~Lomuscio (eds.)}
\copyrightyear{2021}
\acmYear{2021}
\acmDOI{}
\acmPrice{}
\acmISBN{}

\makeatletter  
\def\@copyrightowner{
All rights reserved. This is the authors' version of the work. It
is posted here for your personal use. Not for redistribution. The definitive Version
of Record is published in the aforementioned proceedings.
}
\makeatother

\usepackage{balance} 
\usepackage{tikz}

\usetikzlibrary{
calc,
shapes,
automata,
decorations.pathreplacing,
}

\usepackage[capitalise,noabbrev,nameinlink]{cleveref}

\newcommand{\bbN}{\mathbb{N}}
\newcommand{\valX}{\mathcal{X}}

\newcommand{\indist}{indistinguishability}
\newcommand{\crossp}{synchronous product}
\newcommand{\Crossp}{Synchronous product}
\newcommand{\calA}{\mathcal{A}}

\newcommand{\range}[1]{[ #1 ]}
\newcommand{\enquote}[1]{``#1''}
\newcommand{\defn}[1]{\emph{#1}}

\newcommand{\PPAL}{\ensuremath{\mathrm{PPAL}}}
\newcommand{\PPALstar}{\ensuremath{\mathrm{PPAL}^*}}

\newcommand{\calM}{\mathcal{M}}

\newcommand{\AP}{{AP}}

\newcommand{\obs}[1]{\stackrel{#1}{\leadsto}}

\newcommand{\Reg}[1]{\mathrm{Reg}\left( #1 \right)}

\newcommand{\ann}[2]{\left\{#1 !\right\} #2}
\newcommand{\annstar}[3]{\left\{#1 !\right\}^{#3} #2}
\newcommand{\annres}[2]{#1\{#2 !\}}
\newcommand{\may}[1]{\langle #1 \rangle}
\newcommand{\knows}[1]{[#1]}
\newcommand{\FV}[1]{{FV}(#1)}

\newcommand{\sem}[1]{\left\llbracket #1 \right\rrbracket}
\newcommand{\esem}[1]{\widetilde{\left\llbracket #1 \right\rrbracket}}
\newcommand{\val}[1]{V(#1)}
\newcommand{\bool}{\mathbb{B}}
\newcommand{\trans}[1]{T_{#1}}

\newcommand{\St}{\mathbb{S}}
\newcommand{\upw}{\uparrow\hspace{-0.2em}}
\newcommand{\Lstar}{$L^*$}
\newcommand{\Lp}{L_{\preceq}}

\author{Daniel Stan}
\affiliation{
  \institution{Technical University of Kaiserslautern}
  \state{Germany}}
\email{stan@cs.uni-kl.de}

\author{Anthony W. Lin}
\affiliation{
  \institution{Technical University of Kaiserslautern, MPI-SWS}
  \state{Germany}}
\email{lin@cs.uni-kl.de}

\title{Regular Model Checking Approach to Knowledge Reasoning over 
Parameterized Systems~(technical report)
}

\begin{abstract}
    We present a general framework for modelling and verifying epistemic
    properties over parameterized multi-agent systems that communicate by 
    truthful public announcements. In our framework, the number of agents or 
    the amount of certain resources are parameterized (i.e. not known a 
    priori), and the corresponding verification problem asks whether a given
    epistemic property is true regardless of the instantiation of the
    parameters. 
    For example, in a muddy children puzzle,
    one could ask whether each child will eventually find out whether
    (s)he is muddy, regardless of the number of children. 
    Our framework is regular model checking (RMC) -based, wherein
    synchronous finite-state
    automata (equivalently, monadic second-order logic over words) are
    used to specify the systems. We propose an extension of public announcement 
    logic as specification language. 
    Of special interests is the addition of the so-called iterated public 
    announcement operators,
    which are crucial for reasoning about knowledge in parameterized systems.
    Although the operators make the model checking problem undecidable, we 
    show that this becomes decidable when an appropriate ``disappearance
    relation'' is given. Further, we show how Angluin's 
    L*-algorithm for learning finite automata 
    can be applied to find a disappearance relation, which is guaranteed
    to terminate if it is regular.
    We have implemented the algorithm and apply this to such examples 
    as the Muddy Children Puzzle, the Russian Card 
    Problem, and Large Number Challenge.
\end{abstract}


\begin{document}


\pagestyle{fancy}
\fancyhead{}


\maketitle 

\section{Introduction}
Consider the standard problem of muddy children puzzle in knowledge reasoning
\cite{halpern-book}.
Suppose that there are a total of $N$ children, where $M \in \{1,\ldots,N\}$ of 
them has a mud on their forehead. Each child can observe whether another child
(but not himself) has a mud on their forehead. The muddy children protocol goes
in rounds. At each round, the father declares that there is a muddy child
(i.e. with a mud on their forehead), and asks the children whether they know if 
they are muddy, to which the children can answer yes/no. The announcements made
by the children are observable by other children. After a few rounds (more
precisely $M$ rounds), all children will discover the so-called common knowledge
of which children (including themselves) are muddy and which are not, regardless
of the value of the parameters $M$ and $N$ (e.g. see \cite{halpern-book}). 

The muddy children puzzle can be constructed as a typical
example of a \defn{parameterized verification
problem}~\cite{sasha-book,AMRZ16,KL13,KL16} but with respect to epistemic properties. 
Even though the problem was shown undecidable for a simple safety property
by Apt and Kozen in the 80s~\cite{AK86}, the past twenty years or so have 
witnessed significant progress in the field of parameterized verification (e.g., 
see~\cite{abdulla-survey12,vojnar-habilitation,sasha-book,Zuck-survey,LR21} for
excellent surveys). Researchers resort to either (1) general semi-algorithmic 
techniques that are applicable to general systems, but either without a 
termination guarantee or the method might terminate with a ``don't know''
answer, or (2) restriction to decidable subproblems (e.g. obtained by imposing
certain structures on the parameterized systems). 
More recently, parameterized verification problem was
also considered in the setting of multi-agent systems (e.g., see
\cite{AMRZ16,sasha-book,LP20,KL13,KL16}). Despite this, very little work has been 
done
on parameterized verification problem with respect to epistemic properties,
which is in particular which is applicable in the simple setting of the muddy children
example. This is an extremely challenging problem, while most of the research
focus in parameterized system verification for a few decades has been on simple
safety
properties  (e.g. \cite{vojnar-habilitation,abdulla-survey12,CHLR17,AHH16,NJ13})
and only recently on liveness properties (e.g. \cite{fairy-tale,LR16}).

\paragraph{Summary of Results.} 
We propose a framework for modelling and model-checking epistemic
properties over parameterized multi-agent systems.
Our emphasis in this paper is
on \emph{general semi-algorithmic solutions} that can lend themselves to automatically
solve a variety of interesting examples in knowledge reasoning.
While our semi-algorithm is not guaranteed to terminate in general, we provide
\emph{a general termination condition}, which is proved to subsume examples like Muddy
Children Puzzle, Large Number Challenge, and Russian Card Problem. We detail
our results below.

Firstly, let us recall a standard setting in the finite 
non-\-para\-met\-erized case using \emph{Public Announcement Logic (PAL)}
\cite{DEL-book,Plaza07} (also see \cite{DR07,DHMR06,DRV05}, which provide 
more detailed modelling and a finite-state model checker). The system is 
represented by a finite Kripke
structure, each of whose (binary) accessibility relation $\obs{a}$ (for each
agent $a$) satisfying the 
S5 axioms,
i.e., $\obs{a}$ is an equivalence relation (reflexive, symmetric, and transitive).
That way, $\obs{a}$ can be interpreted as knowledge-indistinguishability by agent
$a$. PAL then is simply a standard modal logic with one accessibility relation
per agent, as well as public announcement modalities $\ann{\varphi}{}$, whereby 
\emph{each} agent learns about $\varphi$. A standard application of the 
public announcement operator is to model the announcement of a child in 
the muddy children protocol, who declares that he knows whether he has a mud
on his forehead. 

To extend the framework to the parameterized setting, there are a few problems.
Firstly, since the Kripke Structure is now infinite (i.e. the union of all
possible instantiations of the parameter), how do we \emph{symbolically
represent} the Kripke Structure? Secondly, a closer look at the solution to the 
muddy children example via PAL (or similar logics)
\cite{DEL-book,Plaza07,halpern-book} suggests that the formula in the logic
is \emph{different} for different numbers of muddy children. For parameterized
verification, it is essential that we have a \emph{uniform} specification for
the epistemic property regardless of the instantiation of the parameters.
We note that generalizations of epistemic logics that can provide such a 
uniform specification do exist (e.g., quantified epistemic logic \cite{BL09},
iterated public announcement~\cite{MM05,GS11}). However,
the resulting logics are not only undecidable, but there are also no known 
semi-algorithmic solutions that would work for interesting examples.

Our framework (see \S\ref{sec:pmodelling}) is in the spirit of \emph{regular model 
checking} 
\cite{abdulla-survey12,rmc-survey,Blum99,BG04,TL10}, wherein a configuration in the
(parameterized) systems are represented by a string over some finite alphabet
$\Sigma$, while a binary relation $\obs{} \subseteq
\Sigma^* \times \Sigma^*$ is represented by an automata over the product
alphabet $\Sigma \times \Sigma$. [The reader could understand a product 
alphabet just like a normal alphabet, where an automaton would synchronously
read a pair $(a,b)$ of symbols at each step.]
The resulting Kripke structures are called
\emph{automatic Kripke structures} \cite{TL10,Blum99,BG04}. 
One benefit of this framework is that one could encode an infinite number
of accessibility relation $\{\obs{i}\}_{i \in \mathbb{N}}$ (one for each
        agent indexed $i=0,1,2,\ldots$), where $\obs{i} \subseteq \Sigma^*
\times \Sigma^*$, as \emph{one single
automaton} representing $\obs{} \subseteq \Sigma^* \times 
\mathbb{N} \times \Sigma^*$. Since a string encoding $s(i)$ of each number $i 
\in \mathbb{N}$ could be given (e.g. $i = 3$ could be represented in unary), 
this automata could run over some product alphabet, e.g., $\Sigma \times 
        \{0,1\} \times \Sigma$.
Second, to reason about knowledge over automatic Kripke Structures, it is
        important to enrich PAL with a few new features: (1) basic string
        reasoning (e.g. whether $b$ occurs at an even position in the string),
        since configurations in the Kripke models are represented as strings
        (2) iterated public announcement operator $\annstar{\varphi}{}{*}$, since 
        in general an unbounded number of public announcements need to be
        made in parameterized systems (e.g. one announcement per child/round
        in the muddy children protocol).

Our key results is as follows. First, in the absence of the iterated public
announcement operators in the input formula, the model checking
        problem in our framework is \emph{decidable with a nonelementary 
        complexity} (see \S\ref{sec:rmc}). Despite
the high complexity, we show that our implementation~\cite{prototype} works well on examples 
like the parameterized version of the Russian Card Problem \cite{DHMR06,vD03}
        (where the total number of cards is not fixed a priori), 
where the tool verifies anonymous communication between two parties of the 
        system could be achieved (see \S\ref{sec:exp}).
        Second, in the presence of the iterated public
announcement operators in the input formula, although the model
        checking problem is in general undecidable (see \S\ref{sec:rmc}), we provide a semi-algorithm
for the problem tapping into Angluin's L* automata learning algorithm
        \cite{angluin,learning-book} (see \S\ref{sec:disrel}). 
To the best of our knowledge, this is the first application of automata 
learning methods to the parameterized model checking of epistemic properties.
        Loosely speaking,
        the learning algorithm will
attempt the computation of the so-called ``disappearance
relation'', that captures the order in which states are discarded during
the announcements and is likely to exhibit regular patterns of the system.
A termination guarantee is provided in this case
(i.e. when the order can be represented by regular
        languages). We implemented the method and show that
        it can successfully verify the parameterized versions of the Muddy 
        Children Protocol and the Large Number Challenge (see \S\ref{sec:exp}).

\section{Preliminaries}
\label{sec:prelim}

We denote $\bbN$ the set of natural numbers, and for $n\in\bbN$,
$\range{n} = \{x\in\bbN~|~ 0 \leq x < n \}$. 

\noindent
\textbf{Automata Background:}
  An \emph{alphabet} is a finite set $\Sigma$. A \emph{word} $w$ over $\Sigma$ is a
  finite sequence $x_0 \hdots x_{n-1} \in \Sigma^n$,
  of \emph{letters} of $\Sigma$, for some length $n$, which is
  is denoted $|w| = n$.
  We write $w[i] = x_i$ for its $i$-th letter ($i\in \range{|w|}$)
  and $\epsilon$ for the empty word of length~$0$.

    A set of words $L$ is called a \emph{language}. It is \emph{regular} if it
  it can be recognized by a regular expression, or equivalently by a
    non-deterministic automaton (e.g. see \cite{Sipser-book}).
  We denote $\Reg{\Sigma}$ the class of
  regular languages over $\Sigma$ and recall the class is closed under concatenation,
  boolean operations, and Kleene star.

  Let $\Sigma_1\subseteq \Sigma'_1$ and $\Sigma_2$.
Regular languages are also preserved by
\emph{\Crossp{}} and \emph{morphism}:

  For $w_1\in \Sigma_1^l$ and $w_2\in\Sigma_2^l$ two words of the same length
  $l\in\bbN$, we write
  $w_1 \otimes w_2$ for the \emph{\crossp{}} word
  $w\in (\Sigma_1 \times \Sigma_2)^l$ such
  that $\forall i\in\range{l},~ w[i] = (w_1[i],w_2[i])$.
  We extend $\otimes$ to languages, by defining, for
  $L_1 \subseteq \Sigma_1^*$
  and
  $L_2 \subseteq \Sigma_2^*$,
  $L_1 \otimes L_2 = \left\{
     w_1 \otimes w_2 ~\middle|~ w_1 \in L_1 \wedge w_2 \in L_2 \cap \Sigma_2^{|w_1|}
  \right\}$

  A \emph{morphism} is any function
  $f: \Sigma_1 \rightarrow \Sigma_2$, we extend $f$ to words
  over $\Sigma_1$ by defining, for any $w\in\Sigma_1^*$,
  $f(w) = f(w[0])\hdots f(w[|w|-1])\in \Sigma_2^*$, then to languages
  over the superset $\Sigma'_1$: for any $L\subseteq (\Sigma'_1)^*$,
  $f(L) = \{f(w)~|~w\in L\cap (\Sigma_1)^*\}$.

Of particular interest, we define \emph{projection morphisms}:\\
given $\Sigma_1\hdots \Sigma_n$, and 
$1\leq i_1 < \hdots i_k \leq n$, we define:
\[\forall (\alpha_i)_{1\leq i\leq n},
    \pi_{(\Sigma_{i_1},-,\hdots,-,\Sigma_{i_k})}(\alpha_1\hdots \alpha_n) =
    (\alpha_{i_1},\hdots \alpha_{i_k})
\]
For example, \crossp{}'s counterparts can be defined as the morphisms
$\pi_{(\Sigma_1,-)}$ and $\pi_{(-,\Sigma_2)}$, \emph{projections}
on the first and second component, respectively.

We encode positions inside a word with the alphabet $\bool=\{0,1\}$
and for $0\leq i < l$, $V(i,l) = 0^{i}10^{l-i-1}\in\bool^l$ encodes the $i$-th position.

When the meaning is clear,
we will at times identify a finite 
automaton
$\mathcal{A}$~and its recognized language \mbox{$\mathcal{L}(\mathcal{A})\in \Reg{\Sigma}$}.
In particular, whenever we claim a language $L$ is regular, a recognizing
automaton may be provided instead.
Whenever $\Sigma=\Sigma_1\times \Sigma_2$, the automaton may also be called
a \emph{length-preserving} transducer, or simply \enquote{transducer},
as it can be interpreted as an automaton
mapping a word $w_1 \in \Sigma_1^*$ to (non-deterministically) a word
$w_2 \in \Sigma_2^*$ of the same length, such that $w_1 \otimes w_2 \in L$.

\section{Our Framework}
\label{sec:pmodelling}

In this section, we provide our regular model checking framework to knowledge
reasoning over parameterized systems. The section has two parts. First, an 
extension of 
PAL called \PPAL{} (Parameterized PAL) that is interpreted over a parameterized
Kripke structure. Second, a regular presentation of parameterized Kripke
structure, over which \PPAL-model checking is decidable.

\subsection{Parameterized Public Announcement Logic}
The logic \PPAL{} will be evaluated on a parameterized Kripke structure. Loosely,
such a structure represents a parameterized system, which can be viewed as a
union of an 
infinite family of structures, each obtained by instantiating the parameter. 
Each state will be assigned a fixed parameter instantation, shared by all its
successors.
For simplicity, we use only one parameter called the \emph{state size}, which 
quantifies the (maximal) number of agents involved, as well as the number of 
copies of atomic propositions.

\begin{definition}
  \label{def:pkrike}
    A \emph{parameterized Kripke structure} is a tuple\\
  \mbox{$\calM = (S,\AP, \obs{}, L, |\cdot|)$}
  where:
  \begin{itemize}
    \item $S$ is a (possibly infinite) set of states;
    \item $\AP$ is a finite set of atomic propositions;
    \item $|\cdot|$ maps any state $s\in S$ to its size $|s|\in\bbN$;
    \item $L$ maps any state $s\in S$ and index
        $i\in \range{|s|}$
        to its \emph{labelling} $L_i(s)\subseteq \AP$;
    \item $\obs{}\subseteq S\times \bbN \times S$ is a $\bbN$-labelled
        accessibility relation between states, called \emph{\indist{} relation},
        such that any triple $(s,i,s')\in \obs{}$ satisfies
        $0 \leq i < |s|=|s'|$. We assume: for any $s\in S$
        and $0\leq i < |s|$, we have $(s,i,s) \in \obs{}$.
  \end{itemize}
\end{definition}

$(s,i,s')\in \obs{}$ is written $s\obs{i} s'$ and reads
"if $s$ is the actual state of the system (world), the $i$-th agent
entertains the possibility that the current state is actually $s'$,
given its observation."
Even though this is not enforced by our definition, most of the proposed
models below will assume $\obs{i}$ to be an equivalence relation, for all $i$,
and this property will be preserved when deriving models.

\begin{example}
  \label{ex:muddy}
  \Cref{fig:muddycube} depicts a parameterized Kripke structure for the
  muddy children puzzle, where $S=\{m,c\}^*$, $\AP=\{m\}$, and
  the size $|w|$ of a state $w\in S$ is defined as its length.
  For all $i\in \range{|w|}$, $L_i(w) = \left\{
    \begin{aligned}
        \{m\}&\text{~if $w[i] = m$}\\
        \emptyset&\text{~otherwise}
    \end{aligned}
  \right.\quad \in 2^{\AP}$
  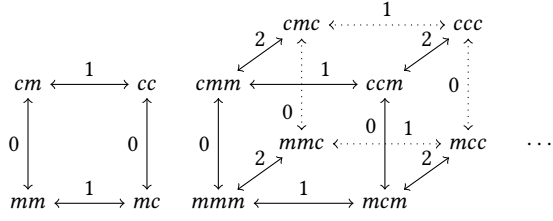
\begin{figure}

     \begin{tikzpicture}
       \tikzset{rn/.style={node distance=5em}}
       \tikzset{rnw/.style={rn, node distance=5em}}
     
       \node[rn]                (mm) {$mm$};
       \node[rnw, right of=mm] (mc) {$mc$};
       \node[rn , above of=mm] (cm) {$cm$};
       \node[rn , above of=mc] (cc) {$cc$};
       \draw[<->] (mm) edge node[above] {$1$} (mc)
                       edge node[left]  {$0$} (cm)
                  (mc) 
                       edge node[left]  {$0$} (cc)
                  (cm) edge node[above] {$1$} (cc);
     
       \tikzset{rnw/.style={rn, node distance=7em}}
     
       \node[rn, right of=mc, node distance=3em] (mmm) {$mmm$};
       \node[rnw, right of=mmm] (mcm) {$mcm$};
       \node[rn , above of=mmm] (cmm) {$cmm$};
       \node[rn , above of=mcm] (ccm) {$ccm$};
       
       \node[rn] at ($(mmm)!0.5!(ccm)$) (mmc) {$mmc$};
       \node[rnw, right of=mmc]         (mcc) {$mcc$};
       \node[rn , above of=mmc]         (cmc) {$cmc$};
       \node[rn , above of=mcc]         (ccc) {$ccc$};
     
       \draw[<->] (mmm) edge node[above] {$1$} (mcm)
                        edge node[left]  {$0$} (cmm)
                        edge node[above] {$2$} (mmc)
                  (mcm) 
                        edge node[left,pos=0.7]  {$0$} (ccm)
                        edge node[above] {$2$} (mcc)
                  (cmm) edge node[above,pos=0.7] {$1$} (ccm)
                        edge node[above] {$2$} (cmc)
                  (ccm) 
                        edge node[above] {$2$} (ccc);
       \draw[<->,dotted]
                  (mmc) edge node[left,pos=0.2]  {$0$} (cmc)
                        edge node[above,pos=0.7] {$1$} (mcc)
                  (ccc) edge node[left]  {$0$} (mcc)
                        edge node[above] {$1$} (cmc)
                        ;
       \node[right of=mcc, node distance=3em] {$\hdots$};
     \end{tikzpicture}
    \caption{First members of the parameterized Kripke family of the Muddy
    children example, with parameter~$2$~(left) and~$3$~(right),
    self loops are omitted.}
    \label{fig:muddycube}
    \Description{With 2 agents, the Kripke structure has a square shape, when
    two nodes are attached if they differ by only one letter. The structure
    became a cube in dimension/parameter 3.}
  \end{figure}
\end{example}

\begin{definition}
  \label{def:ppal}
    We define a formula $\varphi$ in \emph{parameterized public announcement 
    logic} (\PPAL)
  by the following grammar:
  \[
    \varphi ::=
        \top~|~
        \varphi\wedge \varphi~|~
        \neg \varphi~|~
        \exists i: \varphi~|~
        i = 0~|~
        i\%k = 0~|~
        i = j+k~|~
        p_i~|~
        \may{i}\varphi~|~
        \ann{\varphi}{\varphi}
 \]
 Where
    $i,j$ are index variables,
    $k\in\mathbb{N}$ is any integral constant and
    $p\in{\AP}$ is any atomic proposition.
\end{definition}
Intuitively, \PPAL{} extends PAL by an indexing capability, so
that one could easily refer to the $i$th agent in the system. 
This is to some extent akin to how indexed LTL extends LTL \cite{sasha-book}. 
However, we also suitably restrict the indexing capability (essentially,
the difference between the indices of two agents is a certain constant $k$, or
that the index of agent is $k \pmod{d}$ for some constants $k$ and $d$). This
is essentially the extension of the difference logic \cite{KS08} with modulo 
operators. This restriction makes the logic amenable to regular model checking 
techniques, but is also sufficiently powerful for modelling typical examples in
parameterized systems.

\textbf{Shorthands:} 
Boolean connectives
$\vee,\rightarrow,\leftrightarrow$ and universal quantification
$\forall$ can be encoded in a standard way.
The formula $\knows{i}\varphi \equiv \neg \may{i} \neg \varphi$
encodes that agent~$i$ knows with certainty that $\varphi$ holds.
Usage of constants is also allowed:
$i = k\equiv \exists j: j=0\wedge i=j+k$,
$p_k \equiv \exists i: i=k \wedge p_i$,
$\may{k} \varphi \equiv \exists i: i=k \wedge \may{i}\varphi$.

We denote $\FV{\varphi}$ for the set of
(\enquote{not quantified}) \emph{free variables}, of $\varphi$.
We say that $\varphi$ is a closed formula whenever $\FV{\varphi} = \emptyset$.
For any set $X$ of index variables, a function $\mu\in \bbN^{X}$ is called a
valuation. For a valuation $\mu$ and a formula $\varphi$, we write
$\varphi(\mu)$ for the instantiated formula where each occurrence of $x\in X$
has been replaced by $\mu(x)$. In particular, if $\FV{\varphi}\subseteq X$, then
$\varphi(\mu)$ is a closed formula.

\begin{definition}
  \label{def:semantics}
  For a parameterized Kripke structure $\calM$,
  a state $s\in S$,
  a \PPAL{} formula $\varphi$,
  and a valuation $\mu\in\bbN^{\FV{\varphi}}$, we
  define the \emph{satisfaction relation} $\vDash$, inductively, by
  $\calM, s, \mu \vDash \varphi$ if, and only if,
  $\forall i, \mu(i) \in \range{|s|}$ and one of the following condition
  holds:
  \[
  \begin{aligned}
    \varphi &\equiv \top \\
    \varphi &\equiv \psi_1 \wedge \psi
         \text{~and~}
        \calM, s, \mu \vDash \psi_1
        \text{~and~}
        \calM, s, \mu \vDash \psi_2\\
    \varphi &\equiv \neg \psi
         \text{~and~}
        \calM, s, \mu \not\vDash \psi\\
    \varphi &\equiv\exists i: \psi
        \text{~and~}
        \calM, s, \mu' \vDash \psi
        \text{~for some $\mu'$ s.t.~}
        \forall x \neq i, \mu(x) \equiv \mu'(x)\\
    \varphi &\equiv i = j+k
         \text{~and~} \mu(i) \equiv \mu(j) + k\\
    \varphi &\equiv i = 0
         \text{~and~} \mu(i) \equiv 0\\
    \varphi &\equiv i\%k = 0
         \text{~and~} \mu(i)\%k \equiv 0\\
    \varphi &\equiv p_i
         \text{~and~}
        p\in L_{\mu(i)}(s)\\
    \varphi &\equiv \may{i}\psi
         \text{~and
        there exists $t\in S$ such that~}
        s\obs{\mu({i})} t \text{~and~} \calM, t, \mu \vDash \psi\\
    \varphi &\equiv \ann{\psi_1}{\psi_2}
         \text{~and~}
        \calM, s, \mu \vDash \psi_1 \text{~implies~}
        \annres{\calM}{\varphi(\mu)}_\mu, s, \mu \vDash \psi
  \end{aligned}
  \]
  where for any closed \PPAL{} formula $\psi$,
  $\annres{\calM}{\psi}$ is the (parameterized) Kripke structure
  $\calM$ restricted to the state space satisfying $\psi$:
  $\annres{S}{\psi} = \{s~|~\calM,s,\cdot \vDash \psi\}$.
\end{definition}

Note that we adopt here the \emph{vacuous truth} semantics for the
public announcement operator: whenever a state doesn't satisfy a publicly
announced property, it satisfies its conclusion.
This choice will turn out to be more convenient
with our examples involving the newly iterated public announcement.
While an alternative definition $\varphi \wedge \ann{\varphi}{\psi}$ is
possible, they are both expressively equivalent.

It is important to notice that the logic does not make a distinction between
variables designed for atomic propositions manipulation  and variables for
indexing agents. Not only this simplification makes our definition more concise,
it also enables the specification of relationships between agents and \emph{their}
atomic propositions.

\begin{example}
  \label{ex:muddyunique}
  Consider the scenario of the muddy children puzzle, where
  the father announces that there is \emph{exactly} one muddy child.
  \enquote{after this announcement, every child knows their own
  state} is encoded as the formula:
  \[
    \ann{
        \exists i: m_i \wedge
        \forall j, i\neq j \rightarrow \neg m_j
    }{
        \forall i, \knows{i}{m_i} \vee \knows{i}{\neg m_i}
    }
  \]
\end{example}

\subsection{Regular Kripke Structures}

We now provide a regular presentation of parameterized Kripke structures,
and define the model checking problem.
\begin{definition}
  \label{def:regkripke}
  Let
  $\calM=(S,\AP, \obs{}, L, |\cdot|)$
  be a parameterized Kripke structure. It
  is \emph{regular} if there exists an alphabet $\Sigma$
  such that:
  \begin{itemize}
    \item $S\subseteq \Sigma^*$;
    \item For all $s\in S$, $|s|$ is the actual length of $s$, seen as a word;
    \item For all $i \in \range{|s|}$, $L_i(s) = L_0(s[i])$;
    \item The \indist{} relation can be encoded as a transducer, more precisely
        the following language is regular:
        \[
        \trans{\calM} = \left\{
            s \otimes \val{i,|s|} \otimes t~\middle|~
            s \obs{i} t
        \right\}
        \]
  \end{itemize}
\end{definition}
Recall that we assume the reflexivity $\obs{i}$, for each $i \in \mathbb{N}$.
Hence, the state space $S$ of a regular Kripke structure
is also regular, since $S = \pi_{\Sigma,-}(\trans{\calM})$ is a morphism
image.
In the rest of the paper, we will assume the labelling $L$ to be fixed, and
identify any regular Kripke structure $\calM$ with its regular language
$\trans{\calM}$, seen as a transducer. The following proposition justifies
the validity of the above restriction.
\begin{proposition}
    Given an \indist{} relation $(\obs{i})_i$, encoded as a transducer,
    checking any of the following properties to be satisfied by $\obs{i}$ (for each
    $i \in \mathbb{N}$) is decidable:
    (1)~reflexive, (2)~symmetric, and (3)~transitive.
\end{proposition}
This follows from the fact that reflexivity, symmetry, and transitivity of
a binary relation are first-order decidable, and that first-order model
checking over regular Kripke structures (more generally 
\emph{automatic structures}) is decidable \cite{Blum99,BG04}. As a remark,
it follows also that checking whether a regular Kripke Structure satisfies the S5 axioms
(whether all $\obs{i}$ are equivalence relations) is decidable.

\begin{example}[Muddy children]
  \label{ex:muddyreg}
  The parameterized Kripke structure of \cref{ex:muddy} is regular:
  the transducer $\trans{\calM}$
  is recognized by the NFA depicted in \cref{fig:muddydfa}.
  For example, the accepting run
  $
  q_0 \xrightarrow{(c,0,c)} q_0 \xrightarrow{(c,0,c)} q_0 \xrightarrow{(m,1,c)} q_f
  $
  for the word $ccm\otimes 001 \otimes ccc$
  encodes the observation $ccm\obs{2}ccc$.

  \begin{figure}
  \begin{tikzpicture}
    \tikzset{rn/.style={draw, ellipse}}
    \node[rn] (q0) {$q_0$};
    \node[rn,node distance=15em, right of=q0, accepting] (q1) {$q_1$};

    \draw[->] (q0) edge[loop left]
                   node[left,yshift=-0.5em] {\footnotesize $(m,0,m)$}
                   node[left,yshift= 0.5em] {\footnotesize $(c,0,c)$}
                        (q0)
                   edge[<-] +(-2em,+2em)
                   edge[in=180,out=0]
                     node[above] {$(m,1,m),(m,1,c)$}
                     node[below] {$(c,1,m),(c,1,c)$}
                        (q1);
    \draw[->] (q1) edge[loop right]
                   node[right,yshift=-0.5em] {\footnotesize $(m,0,m)$}
                   node[right,yshift= 0.5em] {\footnotesize $(c,0,c)$}
                        (q1);
  \end{tikzpicture}
  \caption{Transducer for the Muddy children}
  \label{fig:muddydfa}
  \Description{
  The transducer starts from an initial control state q0,
  and loops on it while reading triples equating source and target letter,
  the observation bit being set to 0. A transition from
  q0 to the accepting state q1 is allowed when reading any triple with
  observation bit set to 1.
  Further equating triples, with observation bit set to 0,
  are allowed as loops on q1.
  }
  \end{figure}
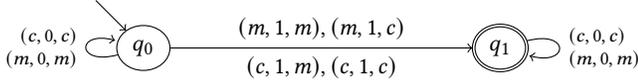
\end{example}

The \emph{regular model checking problem for \PPAL{}} is the problem of model
checking \PPAL{} formulas over regular Kripke structures:
given a regular Kripke
structure $\calM$, and a formula $\varphi$, check if 
the following satisfaction set is empty:
  \[
    \sem{\varphi}(\calM) := \left\{
        (s,\mu)\in S\times \bbN^{\FV{\varphi}}
            ~\middle|~
        \calM, s, \mu \vDash \varphi
    \right\}
  \]
Although we are considering non-pointed Kripke structures, the setting is
not restrictive here, as initial states could be specified by
adding an extra atomic proposition
${init}\in \AP$ and replacing $\varphi$ by $\varphi' \equiv {init} \rightarrow \varphi$.
\section{Regular Model Checking of PPAL}
\label{sec:rmc}

Our main result in this section is the decidability of regular model
checking of \PPAL.
\begin{theorem}
  \label{thm:rmc}
  Given
  a regular Kripke structure $\calM$ and
  a closed \PPAL{} formula $\varphi$, its semantics
  $\sem{\varphi}(\calM)$ is regular and computable.
\end{theorem}

When evaluating a public announcement, the Kripke structure may be modified
in a way that is dependent of the current valuation. The crux of the proof lies
in carrying a family of regular Kripke structures, encoded as a single extended
transducer. The following lemma makes our claim more precise:
\begin{lemma}
  \label{lem:recursive}
  Let $\valX$ be a finite set of variables,
  $\varphi$ a \PPAL{} formula with $\FV{\varphi}\subseteq \valX$, and
    $T\in \Reg{\Sigma\times \bool \times \Sigma \times \bool^{\valX}}$.
  We assume that
    for any $v\in\bool^{\valX}$,
    the transducer $\{w~|~w\otimes v \in T\}$ represents a
    regular Kripke structure denoted $\calM_v$.
  Then, the extended semantics
  \[
    \esem{\varphi}(T)=\{
     s\otimes v ~|~
       \exists \mu\in\bbN^{\valX}:
       v = V(\mu,|s|)
       \wedge
        \calM_v,s,\mu \vDash \varphi
  \}\]
  can
  be recursively computed using boolean, \crossp{}
  and morphism operations on regular languages.
\end{lemma}

\begin{proof}
  {
  \begin{itemize}
    \item $\esem{ \top }(T) = \pi_{(\Sigma,-,-,\bool^{\valX})}(T)$;
    \item 
        $\esem{ \varphi_1 \wedge\varphi_2 } (T) = 
         \esem{\varphi_1}(T) \cap
         \esem{\varphi_2}(T)$;
    \item $\esem{\neg\varphi}(T)  =
        \esem{\top}(T)
        \backslash
        \esem{
        \varphi
        }
        (T)
    $;
    \item
        An existential quantification over $i\in \valX$ is implemented
        by removing the information about $i$'s position.
        For $\alpha=t \otimes v\otimes x\in
        ((\Sigma\times \bool \times \Sigma)\times \bool^{\valX} \times \bool)^*$,
        we define
        $F(\alpha) = t \otimes v[i/x]$.
        Hence,
        $\esem{\exists i: \varphi}(T) = F(\esem{\varphi}(T)\otimes 0^*10^*)$;
    \item $i=j+k$ is encoded by the fixed regular expression:
        \[
        L_k =
        \left\{
        \begin{aligned}
            (0,0)^*(1,0)(0,0)^{k-1}(0,1)(0,0)^*
            &\text{~if $k>0$}\\
            (0,0)^*(1,1)(0,0)^*
            &\text{~otherwise}
        \end{aligned}
        \right.
        \]
        For $w\otimes v\in
            (\Sigma\times \bool\times \Sigma)\times\bool^\valX$,
        we consider the morphism $F$ defined on any
        tuple $\alpha = t\otimes v\otimes (v(i),v(j))$ by
        $F(\alpha) = t\otimes v$, so we
        finally have $
            \esem{i=j+k}(T) =
                F(\esem{\top}(T)\otimes L_k)$;
    \item $\esem{p_{i}}(T) = 
        \pi_{(\Sigma,-,-,\bool^{\valX})}(T)
        \cap
        A^* B A^*
        $,
        where
        \[
        A = 
        \{\alpha\in\Sigma~|~p\notin L(\alpha)\}
        \times
        \{v~|~v(i) = 0\}
        \]
        \[
        B =
        \{\alpha\in \Sigma~|~p\in L(\alpha)\}
        \times
        \{v~|~v(i) = 1\}
        \]
  \end{itemize}
  }
  \begin{itemize}
    \item $\esem{\may{a_i} \varphi }(T) = 
        \pi
        \left(
        T
        ~
        \cap
        ~
        A^*BA^*
        ~\cap~
        (\Sigma\times \bool)^* \otimes \esem{\varphi}(T)
        \right)
    $
    \[
    \text{where:}\quad\left\{
    \begin{aligned}
        A&=\Sigma \times \{0\}\times \Sigma \times
        \{v~|~v(i) = 0\} \\
        B&=\Sigma \times \{1\}\times \Sigma \times
        \{v~|~v(i) = 1\}\\
        \pi&(\alpha,\beta,\gamma,\eta)  = (\alpha,\eta)
     \end{aligned}
     \right.
    \]
    Intuitively, we intersect the transducer with legal moves
    where the current observational player matches the variable~$i$.
    We also intersect
    with the transducer that always ends up in a state and valuation
    satisfying~$\varphi$.
  \item
    The implementation of the public announcement is by far the most complex
    one as, we need first to
    introduce the public announcement transducer
  $T\{\varphi!\}$, encoding for any $v$, the regular Kripke structure obtained
  from $\calM_v$, after announcing $\varphi(\mu_v)$:
  \[
    T\{\varphi!\} = \bigcup_{v\in\bool^\valX}
        \left(\trans{\calM_v\{\varphi(\mu_v)!\}}\right)
        \otimes 
        \{v\}
  \]
 $T\{\varphi!\}$ is actually regular:
 we first build $\esem{\varphi}(T)$ in order to construct a regular
 Kripke on this state space.
 In order to do so, we define the morphism $F$ defined for~
 any\footnote{notice that the same valuation $v$ appears on both sides.}
 $t=w\otimes v \otimes x \otimes w' \otimes v\in (\Sigma\times \bool^\valX
 \times \Sigma\times \bool^\valX)^*$ by
 $F(t) = w\otimes x \otimes w' \otimes v$.
 Then, it remains to intersect the image transducer with the initial model:
 $T\{\varphi!\} = T\cap F(\esem{\varphi}(T) \otimes 0^*10^* \otimes
    \esem{\varphi}(T))$.
  Finally, we conclude with the implementation of the (vacuous truth) semantics
  of the public announcement:
  \[
    \esem{\ann{\varphi}{\psi}} =
        \esem{\neg \varphi}(T)\cup
        \esem{\psi}(T\{\varphi!\})
  \]
  \end{itemize}
\end{proof}

\begin{example}
  \label{ex:muddyregannouncement}
  Consider again the regular Kripke structure of \cref{fig:muddydfa} and
  the effect of publicly announcing "there is at least one muddy child":
  initially $\calM$ has state space $\Sigma^* = \{m,c\}^*$. After
  $\ann{\exists i: m_i}{}$, it is reduced to $\Sigma^*\{m\}\Sigma^*$.
  After announcing "no one knows (s)he muddy", namely
  $\ann{\forall i, \may{i}\neg m_i}{}$, it is further reduced to
  $\Sigma^*\{m\}\Sigma^*\{m\}\Sigma^*$.
  And after $k$ similar announcements, the resulting state space
  becomes $\Sigma^* \left(\{m\}\Sigma^*\right)^k$. This sequence of announcements,
  however, cannot continue forever as each iteration removes all states of
  length $k-1$.
\end{example}

As the reader easily infers, the \PPAL{} logic is suitable for the model
checking of regular Kripke structures of a given size, but cannot keep up in
the parameterized setting, when the number of announcements in the specification
depends on the parameter.

Informally, we would like to embed an arbitrary but finite number of
iterations, namely \emph{iterated public announcement} operator~\cite{MM05}:
    $\underbrace{
        \ann{\varphi}{}
        \ann{\varphi}{}
        \hdots
        \ann{\varphi}{}
    }_{\text{abritrarily~many~times}}
    \psi
    $
\begin{definition}
  \label{def:ppalstar}
  A formula $\varphi$ is in \PPALstar{} if it is in the grammar of
  \PPAL{}, augmented\footnote{
     The construction for a fixed $k\in\bbN$ is only a syntactic
     sugar useful.
  }with
  $\varphi ::= \annstar{\varphi}{\varphi}{k}~|~\annstar{\varphi}{\varphi}{*}$
  with $k\in\bbN$.
  The semantics is given by induction on $k$:
  \begin{itemize}
    \item $\sem{\annstar{\varphi}{\psi}{0}}(\calM) = \sem{\psi}(\calM)$;
    \item $\sem{\annstar{\varphi}{\psi}{k+1}}(\calM) =
        \sem{\annstar{\varphi}{\left(\ann{\varphi}{\psi}\right)}{k}}(\calM)
        $;
    \item $\sem{\annstar{\varphi}{\psi}{*}}(\calM) = 
        \bigcup_{k\geq 0} \sem{\annstar{\varphi}{\psi}{k}}(\calM)$.
  \end{itemize}
\end{definition}

\cref{thm:rmc} ensures that model checking of a regular Kripke
structure against a \PPAL{} formula is decidable, by reduction to regular
language universality problem. However, the translation of a formula
into a regular language may involve several exponential blow-ups, so the overall
running time may become non-elementary. Moreover,
this translation does not apply to the newly introduced $\annstar{\cdot}{}{*}$
operator, and decidability is not guaranteed in this case.
We clarify now these complexity questions:

\begin{theorem}
  \label{thm:decidability}
  There exists a regular structure $\calM$, such that:
  \begin{enumerate}
    \item Model checking against a \PPAL{} formula is non-elementary;
    \item Model checking against a \PPALstar{} formula is undecidable.
  \end{enumerate}
\end{theorem}

\begin{proof}
  \begin{enumerate}
    \item In \cite[Proposition 20]{wt-csl09}, the author constructs an automatic
    structure $\calM$
    whose modal logic theory is non-elementary. Modal logic can be seen
    as a particular fragment of \PPAL{}, with only one agent.
    An automatic structure can also be seen as a
    regular Kripke structure with only one agent. The hardness reduction is
    therefore immediate.

  \item 
    We construct now a regular Kripke structure $\calM$ such that its \PPALstar{}
    theory is undecidable.
    { 
    To this end, we encode the Minsky machine halting problem~\cite{minsky67}:
    a 2-counter (Minsky) machine is a tuple $(Q,q_0,q_f,\delta)$
    where
    \begin{itemize}
      \item $Q$ is a finite subset;
      \item $q_0\in Q$ is the initial state;
      \item $q_f\in Q$ is the final state;
      \item $\delta\subseteq Q\times \{{test},{inc},{dec}\}\times \bool\times Q$
        is the set of transitions.
    \end{itemize}
    The semantics of such a machine is defined over the configuration space
    $Q\times \bbN^2$, with $(q,x_1,x_2) \xrightarrow{t} (r,y_1,y_2)$ if, and only if,
    the following conditions hold:
    \begin{itemize}
      \item $t=(q,{op},i,r)\in \delta$ for some
        \mbox{$({op},i)\in \{{test},{pos},{inc},{dec}\}\times \bool$};
      \item $x_{1-i} = y_{1-i}$;
      \item if ${op}={test}$, then $x_i = y_i = 0$;
      \item if ${op}={pos}$, then $x_i = y_i > 0$;
      \item if ${op}={inc}$, then $x_i +1 = y_i$;
      \item if ${op}={dec}$, then $x_i = y_i+1$;
    \end{itemize}
    We say that such a machine terminates if there exists a finite path from
    configuration
    $(q_0,0,0)$ to configuration $(q_f,x_0,x_1)$ for some $(x_0,x_1)$.

    Moreover, we assume, without loss of generality, our \mbox{2-counter} machines
    to be deterministic, namely: if for any configuration $\gamma$, there exists
    at most one $t\in \delta$ such that $\gamma\xrightarrow{t} \gamma'$ for
    some $\gamma'$.

    We consider the regular Kripke structure $\calM$ with
    $\AP=\{p,c^{(1)},c^{(2)}\}$, and $\trans{\calM}$ the complete transducer
    $\Sigma^* \otimes 0^*10^* \otimes \Sigma^*$.
    
    Given a 2-counter machine $(Q,q_0,q_f,\delta)$, we construct a formula
    $\varphi$ such that the machine terminates if, and only if,
    $\sem{\varphi}(\calM) \neq \emptyset$.
    We assume for our encoding that $Q=\range{|Q|}$, allowing us to encode
    the current state as a unary position.

    Let's first restrict the model to states where each proposition is true
    at exactly one position, by announcing:
    \[\varphi_{m} = \exists i_0,i_1,i_2:
      \left\{\begin{aligned}
        c^{(0)}_{i_0} \wedge
        c^{(1)}_{i_1} \wedge
        p_{i_0} = j
        \\
        \forall j, 
        c^{(0)}_j \rightarrow i_0 = j \wedge
        c^{(1)}_j \rightarrow i_1 = j \wedge
        p_j \rightarrow i_0 = j
      \end{aligned}
      \right.
    \]
    Intuitively, a configuration $(q,x_0,x_1)$ will be encoded
    as the state word
    $V(x_0,l) \otimes V(x_1,l)\otimes V(q,l)$ for any $l\geq {max}(q,x_0,x_1)$.

    We construct now a formula $\varphi_t$ expressing that the current configuration
    still has successor:
    \[
      \varphi_t = \exists i_0,i_1,j\neq q_f:
        c^{(0)}_{i_0}\wedge c^{(1)}_{i_1} \wedge
        p_{j} \wedge 
        \bigvee_{(q,{op},k,q')\in \delta}\hspace{-2em}
            j = q \wedge 
            \may{0} (p_{q'} \wedge \varphi_{{op},k})
    \]
    where
    \[
    \varphi_{{op},k} = \left\{\begin{aligned}
      i_k = 0 \wedge c^{(0)}_{i_0} \wedge c^{(1)}_{i_1} 
        & \text{~when ${op}={test}$}\\
      \exists l: i_k = l +1 \wedge c^{(0)}_{i_0} \wedge c^{(1)}_{i_1} 
        & \text{~when ${op}={pos}$}\\
      \exists l: i_k = l+1 \wedge c^{(k)}_{l} \wedge c^{(1-k)}_{i_{1-k}} 
        & \text{~when ${op}={dec}$}\\
      \forall l, l = i_k +1 \rightarrow c^{(k)}_{l} \wedge c^{(1-k)}_{i_{1-k}}
        & \text{~when ${op}={inc}$}\\
    \end{aligned}\right.
    \]

    Note that a state $s$ encoding the configuration $(q,x_0,x_1)$,
    where a increment of $x_i$ is available
    but $x_i = |s|-1$, is never removed, even though state $s$ has no successor
    in the current Kripke structure: we adopt this convention since
    any state in $(0^{+}\cdot s)$ still has a successor.

    We encode now the whole formula:
    \[
        \varphi = 
        p_{q_0} \wedge c^{(0)}_0 \wedge c^{(1)}_1 \wedge
        \ann{\varphi_m}{
            \annstar{\varphi_t}{\bot}{*}
        }
    \]

    We claim that 
    $\sem{\varphi}(\calM) \neq \emptyset$
    if, and only if,
    the machine terminates.

    \begin{itemize}
      \item If the machine terminates, there exists a finite run\\ 
        $(q_0,x_0, y_0)\hdots
         (q_n,x_{n}, y_{n})$ with $q_n = q_f$ and $x_0 = y_0 = 0$.\\
        We let $l = {max}\{ x_i, y_i ~|~0\leq i \leq n\}$.
        We prove by
        (decreasing) induction on $i$ that
        $V(x_i,l) \otimes V(y_i,l)\otimes V(q_i,l)
            \in \sem{
            \ann{\varphi_m}{
                \annstar{\varphi_t}{\bot}{*}
            }}(\calM)$.
        The result holds for $i=n$ since $\varphi_t$ is not satisfied for $q_f$,
        then, we follow the semantics definition for the induction case:
        by determinacy of the machine, there exists at most one valid instruction
        $(q_{i-1},{op},k,q_i)$. Because no state can satisfy $\bot$, the
        state $V(x_i,l) \otimes V(y_i,l)\otimes V(q_i,l)$ has to eventually be
        removed, hence $\varphi_t$ becomes eventually unsatisfied
        from $V(x_{i-1},l) \otimes V(y_{i-1},l)\otimes V(q_{i-1},l)$.
      \item For the converse implication, assume the machine has an infinite
        run denoted with $(q_i,x_i,y_i)_{i\in\bbN}$.
        For a fixed $l \in \bbN$, we check by induction on $k$, that for all~$i$,
        such that $ l \geq {max}\{x_i,y_i,q_i\}$, we still have
            $\calM, V(x_i,l) \otimes V(y_i,l)\otimes V(q_i,l)\not\vDash 
                \ann{\varphi_m}{\annstar{\varphi_t}{\bot}{k}}
            $.
        \begin{itemize}
            \item For $k=0$, the formula reduces to $\ann{\varphi_m}{\bot}$, and
            since all states are proper encoding of configurations, they satisfy
            $\varphi_m$ but not $\bot$.
            \item Assuming the result is true for some $k\in\bbN$, we prove
            the result still holds for $(q_i,x_i,y_i)$ at step $k+1$, by
            applying either applying the induction hypothesis on\\
            $(q_{i+1},x_{i+1},y_{i+1})$ at step $k$, or checking that
            $x_{i+1}>l$ or $y_{i+1}$ meaning
            $V(x_i,l) \otimes V(y_i,l)\otimes V(q_i,l)$ satisfies $\varphi_t$
            through its increment clause.
        \end{itemize}
        That is to say:
        $\sem{\varphi}(\calM)\cap \Sigma^l = \emptyset$.
        Since the result holds for any $l\in\bbN$, we conclude that
        $\sem{\varphi}(\calM)= \emptyset$.
    \end{itemize}
    } 
 \end{enumerate}
\end{proof}

\section{Disappearance Relation}
\label{sec:disrel}

We study in this section the limit behaviour induced by an operator restricting
incrementally the state space. This study is motivated by the \PPALstar{} construction
$\annstar{\varphi}{\bot}{*}$, whose semantics can be seen as the set of states
\emph{not} being removed, after arbitrarily many state space restriction operated
by the public announcement $\ann{\varphi}{}$. 

For the rest of this section, we fix a more general setting with:
\begin{itemize}
  \item An alphabet $\Sigma$;
  \item An initial state space $\St = \St_0 \subseteq \Sigma^*$;
  \item A function
  $F: 2^{\Sigma^*} \rightarrow 2^{\Sigma^*}$
  restricting the state space, that is to say:
    $\forall X\subseteq \Sigma^*, F(X) \subseteq X$.
  \item For all $k\in \bbN$, we let $\St_{k+1}=F(\St_k)$ and
    $\St_\infty = \cap_{k\geq 0} \St_k$.
\end{itemize}

Moreover, we assume $\St$ to be a regular language, and $F$ to preserve
regular languages, in a way that will later be clarified.

Our study aims at computing the limit set of states $\St_\infty$. Despite
our assumptions, this set is not in general regular nor computable,
as one can observe
as a consequence of the undecidability result of \cref{thm:decidability},
or the following counterexample.
\begin{example}
  \label{ex:counting}
  Consider $\St_0 = \{a\}^*\{b\}^*$, and for all $X\subseteq \St_0$,
  define $F(X)=\{a^nb^m~|~n=m=0 \vee (nm>0 \wedge a^{n-1}b^{m-1} \in X)\}$.
  Then, for any $k\in\bbN$,
  $\St_k = \{a^ib^i~|~0\leq i < k\}\cup \{a^{k+n}b^{k+m}~|~n,m\geq 0\}$
  is regular, but not its limit $\St_{\infty}=\{a^ib^i~|~0\leq i\}$.
\end{example}

\subsection{Unique Characterization}
Let us first remark that the application $F$ is not necessarily
monotone. Consider for example the announcement
$\forall i, m_i \vee \knows{i}\neg m_i$ which reads:
\begin{center}\enquote{
    every non-muddy child knows he's not muddy.
}\end{center}
Then $F(\{cm\}) = \{cm\}$ but $F(\{cm,mm\})=\emptyset$.
As a consequence, $\St_\infty$ is a fixed point of $F$,
but cannot be characterized as the smallest nor the greatest one.
Hence, we narrow down our computation goal by introducing the
following pre-order over states:
\begin{definition}
  \label{def:disrel}
  The \emph{disappearance} relation $\preceq$ is defined
  for every $(s,t) \in \St^2$, by:\[
    s\preceq t
    \quad\text{~if, and only if,~}\quad
        \forall k\in\bbN,
        s\in \St_k
         \Rightarrow
        t \in \St_k
    \]
\end{definition}
Intuitively, $s\preceq t$ means that $s$ disappears from the state space
before~$t$.
$\preceq$ is a total
pre-order, i.e. any two elements are comparable,
the relation is reflexive, transitive, but not necessarily antisymmetric.
Notice that $\St_\infty$ can be characterized as the set of maximal elements
of $\preceq$.

In order to reason over set of states induced by a pre-order, we introduce
the following notations:
\begin{definition}
  \label{def:upweq}
  For a relation $R\subseteq \St\times \St$, and any $s\in \St$, we define
  the \defn{upward-closure} and \defn{equivalence class} of~$s$ by
        $\upw_R s = \{u\in \St~|~(s,u)\in R\}$,
        and
        $[s]_R = \{u\in \St~|~(s,u)\in R \wedge (u,s)\in R\}$,
        respectively.
  We omit the subscript notation when $R=\preceq$. 
\end{definition}
As suggested by its name, the latter notion involves an equivalence relation,
namely $R\cap R^{-1}$ which relates states of $\St$ disappearing at the same
iteration.
On the other hand, the upward-closure $\upw s$ can be interpreted as one
of the iterated $\St_k = \upw s$ for some $k\in\bbN\uplus\{\infty\}$.
When $k<\infty$,
we know this is the \emph{last} iteration before $s$ and all its
equivalent states got removed.
This entails $\St_{k+1} = \St_{k}\backslash [s]$, hence
$F(\upw s) = \upw s \backslash [s]$.
When $k=\infty$, we know on the contrary, that $s$ never disappears, which
also means $[s] = \upw s = \St_\infty$.

\begin{figure}

      \begin{tikzpicture}
        \tikzset{rn/.style={draw, ellipse}}
      
        \node[rn] (q0) {$\St_0 \backslash \St_1$};
        \node[rn, anchor=east]
                  at ($(q0.west)+(\columnwidth,0)$)
                          (q3) {$\St_{\infty}$};
        \node[rn] at ($(q0)!0.4!(q3)$) (q1) {$\St_i\backslash \St_{i+1}$};
        \node[rn] at ($(q1)!0.4!(q3)$) (q2) {$\phantom{\St_{i+1}}$};
        \node at (q2) {\footnotesize $\St_{i+1}\backslash\St_{i+2}$};
      
        \path (q0) edge[draw=none]
                      node[pos=0.25] {$\preceq$}
                      node[pos=0.5] {$\hdots$} 
                      node[pos=0.75] {$\preceq$}
                      (q1);
        \path (q1) edge[draw=none] node {$\preceq$} (q2);
        \path (q2) edge[draw=none]
                      node[pos=0.25] {$\preceq$}
                      node[pos=0.5] {$\hdots$} 
                      node[pos=0.75] {$\preceq$}
                      (q3);
        
        \draw [decorate,decoration={brace,amplitude=10pt,aspect=0.75},
        transform canvas={yshift=-1.9em}]
        (q3.east) -- (q1.west) node [black,pos=0.75,below,yshift=-0.7em,xshift=0.6em]
      {\footnotesize $\St_i = \upw s$};
        \draw [decorate,decoration={brace,amplitude=10pt,aspect=0.75},
        transform canvas={yshift=-0.6em}]
        (q3.east) -- (q2.west) node [black,pos=0.75,below,yshift=-0.7em, xshift=1.3em,anchor=north]
      {\footnotesize $\St_{i+1} = F(\upw s)$};
      
      \node (exs) at
          ($(q0)+(3em,-2.9em)$)
          {$s$};
      \node (ext) at ($(exs)+(2em,0)$) {$t$};
      \path (exs) edge[draw=none]
          node[yshift= 0.2em] {$\prec$}
          node[yshift=-0.2em] {$\succ$}
          (ext);
      \draw ($(q1)+(-2em,-0.3em)$) edge[bend right] (exs);
      \draw ($(q1)+(-0.8em,-0.6em)$) edge[bend right] (ext);
      
      \end{tikzpicture}
  \caption{Hierarchy of the equivalence classes of the disappearance relation}
  \label{fig:hierachy}
  \Description{
    The image describes a series of blocks, ordered by disappearance relation.
    All blocks except the last one correspond to the state space at some
    finite iteration, minus the one at next iteration. Last block represents
    the limit of all iterations, remaining states. Two states in the same
    block are equivalent meaning they both greater/smaller than each other.
    Given a state in a block, taking the upward closure describes all states
    in his block and all upper blocks up to infinity.
    Applying the operator on this upward closure removes the smaller
    block from the description, keep all upper blocks, up to infinity.
  }
\end{figure}
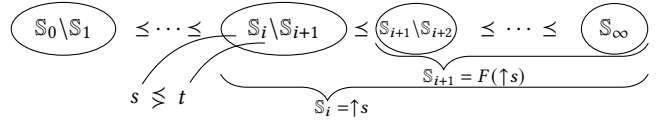
This setting is summarised in the following~\cref{fig:hierachy} and~\cref{prop:orderchr},
the latter also provides a unique characterization under certain
conditions:
\begin{proposition}
  \label{prop:orderchr}
  Let $R\subseteq \St \times \St$.
  If $R = \preceq$, then:
  \begin{enumerate}
    \item $R$ is a total pre-order on $\St$, and
    \item for $s \in \St$,
        $[s]_R = {\uparrow}_R s\backslash F({\uparrow}_R s)$
        or
        $[s]_R = {\uparrow}_R s = F({\uparrow}_R s)$.
  \end{enumerate}
  Moreover, the converse holds whenever $\St$ is finite.
\end{proposition}

\begin{proof}[Sketch]
  A proof of the direct implication being already sketched above, we focus
  on the converse implication:\\
  Assuming that $R$ satisfies the above conditions and $\St$ is finite,
  we prove that $R=\preceq$, by induction on $|\St|$.

  The result is trivial when $|\St|=0$, consider now $\St\neq \emptyset$
  and some
  minimal $s\in \St$ with respect to the total pre-order~$R$, that is to say
  ${\uparrow}_R s = \St$.
  \begin{itemize}
    \item If
      $[s]_R = {\uparrow}_R s\backslash F({\uparrow}_R s)$,
      then $[s]_R = \St_0 \backslash F(\St_0) = \St_0 \backslash \St_1$.
      In particular, $s\notin F(\St_0) = \St_1$.
      Consider $\preceq' = \preceq \cap \St_1\times \St_1$
      and $R'= R\cap \St_1\times \St_1$.
      We easily check
      that $\preceq'$ is the disappearance relation of $F$, with initial set
      $\St_1$, and $R'$ is a total pre-order on $\St_1$ such that for every
      $s\in \St_1$,
        $[s]_R = [s]_{R'}$ and $\upw_{R} = \upw_{R'}$.
      We can therefore apply the induction hypothesis on $\St_1$:
      $\preceq $ and $R$ coincide on $\St_1\times \St_1$.
      
      It remains to show that this is also the case over
      $\St_0\times\St_0 \backslash (\St_1 \times \St_1)$:
      for $u\in \St_0\backslash \St_1=[s]_R$, and $t\in\St$,
      $(u,t)\in R$ (by equivalence to $s$) and $u \preceq t$ since $u\notin \St_1$
      and $t\in\St_0$.
      On the other hand, by minimality of $s$,
      $(t,u)\in R \Leftrightarrow t\in [s]_R \Leftrightarrow
      t \preceq s$.
      We conclude in this case that $\preceq = R$.
    \item If
      $[s]_R = {\uparrow}_R s = F({\uparrow}_R s)$,
      then $[s]_R = \St_{\infty}$ and $R = \preceq = \St\times \St$.
  \end{itemize}
\end{proof}

\subsection{Learning Procedure}
\label{ssec:learnproc}

The unique characterization of \cref{prop:orderchr} paves the way to a
learning procedure for computing~$\preceq$.
More precisely, we consider for this section an encoding of $\preceq$ seen as
as a language over pairs of letters:
    $
    \Lp =
    \left\{
        s \otimes t
        ~\middle|~
        |s| = |t| \wedge s \preceq t
    \right\}
    \subseteq
    \left(\Sigma \times \Sigma\right)^*
$\\
Assuming $\Lp$ is a regular language, we will develop a learning procedure
to construct it.
On the one hand, notice
that this definition of $\Lp$ looses some information about $\preceq$
as it can only relate states of the same length.
On the other hand, this restriction is not crucial as the $\PPAL{}$ logic
is exclusively based on length-preserving transducers. We keep the following
requirement:
\begin{center}
    $(R1)$: {\enquote{$F$ is length-preserving}}
\end{center}
Strictly speaking,
we assume $\Lp$ to be the representation of the family
$(\preceq_k)_{k\in\bbN}$, where for each $k$, $\preceq_k$ is the disappearance
relation starting from the initial state space $\St_0 \cap \Sigma^k$.
As $\Sigma^k$ is finite for a given $k$, this restriction further allows
us to provide a unique characterization of $\Lp$, as provided by
\cref{prop:orderchr}.

We now introduce the \Lstar{} algorithm from Angluin, which allows us 
to learn a finite automaton $\calA$, or equivalently a target
regular language $L_t\in\Reg{\Sigma_t}$, based on queries answered by
an Oracle.
Such an Oracle has to answer \emph{membership} and
\emph{equivalence} queries, by direct access to the target language
or by indirect means.

We explain the exact semantics of \Lstar{} queries for a target
language language~$L_t\in \Reg{\Sigma_t}$, and how they
are answered in this learning procedure, where the target language
is~$\Lp\in\Reg{\Sigma\times \Sigma}$:
\begin{itemize}
  \item \textbf{Membership Queries:} the Oracle is asked 
    whether a given word $w\in \Sigma_t^*$ is in the target language $L_t$.\\
    \textbf{Answer:} we let $s,t \in \Sigma^{|w|}$ with $s\otimes t=w$
    and decide whether $s\preceq t$.
    We proceed to the iterative computation of the sets $(\St_k)_{k\in\bbN}$ and stop
    whenever $s$ or $t$ is no more in the set.
    This is however a semi-decision procedure as it may fail in the case
    where neither $s$ nor $t$ disappear ($s,t\in \St_\infty$).
    To circumvent this issue, we perform the computation on the restricted
    state space of a fixed length $|w|$, namely
    $\St_k\cap \Sigma^{|w|}$, ensuring a finite cardinality.
    As soon as $s,t\in \St_k\cap \Sigma^{|w|}= \St_{k+1}\cap \Sigma^{|w|}$,
    we conclude that $s\preceq t$.
    This leads to our second requirement:
    \begin{center}
        $(R2)$: \enquote{$F$ restricts the state space
    independently for different state sizes.}
    \end{center}
  \item \textbf{Equivalence Queries:}
    Given a candidate language $L$, the Oracle is asked whether
    $L=L_t$ and if not, provides a counterexample $w\in L\backslash L_t \cup L_t \backslash L$.\\
    \textbf{Answer:} we use of \cref{prop:orderchr}, which can be seen
    as a first order characterization of $\preceq$, and translate the listed
    criteria into equivalence problems over regular languages.
    If one regular language equivalence fails, we have to provide a counter
    example to the learning procedure. Unfortunately, a counterexample
    to a criterion of \cref{prop:orderchr} does not directly
    provide a counterexample for \Lstar. For example, a counterexample for
    the transitivity property would consist in a triple
    $(s_1,s_2,s_3) \in \St^3$ such that $s_1\otimes s_2 \in L$,
    $s_2 \otimes s_3\in L$ but $s_1\otimes s_3 \notin L$, and it wouldn't be clear
    whether the property fails because either $(s_1,s_2)$ or $(s_2,s_3)t$
    should be removed from $L$
    or because $(s_1,s_3)$ should be added.
    Nonetheless, since a counterexample was provided for a fixed length $l$,
    $L$ restricted to $(\Sigma\times \Sigma)^{l}$ is not
    r proper encoding of \mbox{$\preceq\cap \Sigma^l \times \Sigma^l$}. A direct
    enumeration of the sequence $(\St_k\cap \Sigma^l)_{k\geq 0}$ will therefore terminate
    and therefore will provide a counterexample.
\end{itemize}

\subsection{Effective and Uniform Regularity}
In order to effectively implement the procedure, we provide the following
equivalent characterization of \cref{prop:orderchr}, in terms
of first-order formulae.
\begin{proposition}
  \label{prop:implemchr}
  Let $R\subseteq \Sigma^*\times \Sigma^*$ and $k\in\bbN$.\\
  $R\cap (\Sigma^k\times \Sigma^k) \neq \preceq_k$
  if, and only if, \emph{any} one of the conditions holds:
  \begin{enumerate}
    \item $\exists s,t: (s,t)\in R \wedge (s,t)\notin \St\times \St$;
    \item $\exists s: (s,s)\notin R$;
    \item $\exists s_1,s_2,s_3: (s_1,s_2)\in R\wedge (s_2,s_3)\in R \wedge
        (s_1,s_3)\notin R$;
    \item $\exists s,t: (s,t)\notin R\wedge (t,s)\notin R$;
    \item \[
    \exists s,t_1, t_2:
        \left\{\begin{aligned}
            (s,t_1) \in R  &\wedge (s,t_2)\in R \\
            (t_1,s)\notin R &\not\leftrightarrow t\in F({\uparrow}_R s) \\
            (t_2,s)\notin R &\vee t_2\notin F({\uparrow}_R s)
       \end{aligned}
       \right.
       \]
  \end{enumerate}
  Where all quantifications are made over $\Sigma^k$.
\end{proposition}

{
\begin{proof}
    Property~$(1)$ enforces $R\subseteq \St\times\St$ while
    properties~$(2)-(4)$ encode respectively reflexivity, transitivity and
    totality, as stated by \cref{prop:orderchr}, after
    taking the negation.

    We provide here a proof of property~$(5)$ built on top on the second
    property
    not being fulfilled.
    Recall first that $F(X)\subseteq X$ for all $X$,
    and $[s]_R \subseteq \upw_R s$ for any $s$, hence
    condition $[s]_R = \upw_R s = F(\upw_R s)$ is equivalent
    to $\upw_R s \subseteq [s]_R \cap F(\upw_R s)$.\\
    After taking the negation, the second property
    of~\cref{prop:orderchr} becomes:
    $
        \exists s:
        [s]_R \neq \upw_R \backslash F(\upw_R s) ~\wedge~
        \upw_R s \not\subseteq [s]_R \cap F(\upw_R s)
    $
    which is equivalent to:
    \[
        \exists s,t_1,t_2:
        \left\{\begin{aligned}
            ((s,t_1)\in R \wedge (t_1,s)\in R)
            \not\leftrightarrow
            ((s,t_1)\in R \wedge t_1 \notin F(\upw_R s))
        \\
            (s,t_2)\in R \wedge
            \left(
            (t_2,s)\notin R \vee
            t_2 \notin F(\upw_R s)
            \right)
        \end{aligned}\right.
    \]
    Hence, after factoring by $(s,t_1)\in R$:
    \[
        \exists s,t_1,t_2:
        \left\{\begin{aligned}
            (s,t_1) \in R &\wedge 
            ((t_1,s)\in R
            \not\leftrightarrow
            t_1 \notin F(\upw_R s))
        \\
            (s,t_2)\in R &\wedge
            \left(
            (t_2,s)\notin R \vee
            t_2 \notin F(\upw_R s)
            \right)
        \end{aligned}\right.
    \]

\end{proof}
}

Based on this first-order characterization, we provide an actual
implementation of equivalence queries on the candidate language $L$,
by resorting to queries on length-preserving transducers, namely
regular languages over $\Sigma\times\Sigma$.
For example, Property~$(1)$ is translated to the query
    $
    L \cap \overline{\St\otimes \St} \stackrel{?}{=} \emptyset
$.

While the predicates $\St\times \St$ and $R$ can be encoded as
the regular languages $\St \otimes \St$ and $L_c$, respectively,
property $(5)$ involves
the computation of the operator $F$ as the following binary predicate:
\[
  F(\upw_R \cdot)=
  \{
    (s,t)~|~s\in F(\upw_R t)
  \}
\]
This condition
is introduced as the last requirement:
\begin{center}
  $(R3)$: \enquote{$F$ is effective and uniformly regular}
\end{center}
Conditions $(R1)-(R3)$ are formally defined through the following conditions:
\begin{definition}
  \label{def:lenuregfun}
  Let $G$ be a function from $2^{\Sigma^*_1}$ to $2^{\Sigma^*_2}$.
  \begin{itemize}
    \item $G$ is \emph{independently length-preserving} if:\\
        $\forall l\in \bbN~\forall X\subseteq \Sigma_1^*,
            G(X\cap \Sigma_1^l) = G(X)\cap \Sigma_2^l$;
    \item $G$ is
        \emph{effectively uniformly regular} if:\\
        For any given alphabet $\Sigma'$
        and $L\in\Reg{\Sigma'\times \Sigma_1}$,
        the following language is
        regular and computable:
        \[
            \left\{
                w'\otimes w_2
            ~\middle|~
                \exists w_1\in\Sigma^{|w_2|}:
                w_2 \in G\left(\left\{
                    w_1~|~w'\otimes w_1 \in L  
                \right\}
            \right)
            \right\}
        \]
  \end{itemize}
\end{definition}

\begin{theorem}
  \label{thm:algo}
  Assume $F$ is an independently length-preserving and uniformly regular
  function.\\
  Then the \Lstar{} learning procedure described in
  \cref{ssec:learnproc} eventually terminates and returns $\Lp$
  if, and only,
  it is regular.
\end{theorem}
{
\begin{proof}
    Thanks to the length-preserving property of $F$, the relation
    $\preceq_k = \Sigma^k \times \Sigma^k \cap \preceq$ coincide with the
    disappearance relation initiated from
    $\St_0\cap \Sigma^k$ with the same operator~$F$.

  Uniform and effective regularity enables the effecive implementation of
  the Oracles:
  \begin{itemize}
    \item First of all, membership queries, as well as counterexample
        generation in the equivalence queries, require the computation
        of sequences $F^i(\St\cap \Sigma^l)$ for different values of $i,l\in \bbN$.
        This can be seen as a weaker form of effective regularity,
        satisfied by the operator.
    \item Equivalence queries implementation relies on the translation of the
        conditions provided by \cref{prop:implemchr} into regular queries
        on transducers over $\Sigma\times \Sigma$.
        Last condition in particular,
        requires, for a candidate relation $R$, encoded as a transducer
        $R\in\Reg{\Sigma\times \Sigma}$, the computation of
        \[\{s\otimes t~|~ t\in F( \uparrow_R s)\}=
            \left\{
                s\otimes t
            ~\middle|~
                \exists u\in\Sigma^{|w_2|}:
                t \in F\left(\left\{
                    u~|~s\otimes u \in L  
                \right\}
            \right)
            \right\}
         \]
  \end{itemize}

  We conclude with the termination guarantees:
  \begin{itemize}
    \item If $L_\preceq$ is regular, the \Lstar procedure terminates
    in polynomial time\cite{angluin}.
    \item Conversely, if the procedure terminates, the returned language $L$
    is regular and passed the equivalence query. Therefore,
    it satisfies the $L_\preceq$ characterization provided
    by \cref{prop:implemchr}, so $L_\preceq = L$ is regular.
  \end{itemize}
\end{proof}
}

\subsection{Application to \PPALstar}
We finally address the general case with the following observation:
\begin{proposition}
  \label{prop:starpreceq}
  Let $\varphi$ and $\psi$ be two closed formula and $\calM$ a parameterized Kripke
  structure, whose state space is $\St$.
  For any $X\subseteq \St$, we define $\calM_{|X}$ for the parameterized
  Kripke structure restricted to $X$ and consider
  the resulting disappearance relation $\preceq\subseteq \St^2$.

  We have:
  \[
    \sem{\annstar{\varphi}{\psi}{*}}(\calM) =
    \left\{
        s\in \St~\middle|~
        \exists t\in \St: s\in
           \sem{\psi}(\calM_{|\uparrow_{\preceq} t}) 
    \right\}
  \]
\end{proposition}

We can easily see that the above set is regular if $\calM$ and $\preceq$ are
both regular. In order to proceed to their computation, we need to provide
the following uniformly regular property:
\begin{proposition}
  Let $\varphi$ be a closed formula on a regular Kripke structure~$\calM$. 
  The application, $F_{\varphi}$ defined by
  \[
    \forall X\subseteq \St,
    F_{\varphi}(X) = 
     \sem{\varphi}(\calM_{|X})
  \]
  is length-preserving, effectively and uniformly regular.
\end{proposition}

\begin{proof}
  Given $L\in \Reg{\Sigma'\times \Sigma}$, we define a new
  regular Kripke structure $\calM'$ storing the information about $\Sigma'$
  in its state space. The construction of $\calM'$ is effective, and
  by \cref{thm:rmc}, we can compute
  $\sem{\varphi}(\calM')$.
\end{proof}
\section{Experiments}
\label{sec:exp}

We developed a prototype tool implementation, using
the Java libraries Learnlib and Automatalib~\cite{IHS15}. Three different
models were specified then
verified\footnote{Experiments were conducted on a i7-8550U CPU @ 1.80GHz machine
with 16GB of RAM and JavaSE-1.8.
The prototype and models are available online~\cite{prototype}.} showing
tractability of the procedure:
\begin{tabular}{ccc}
 Model & Duration & Memory Usage\\
 \hline
 Russian cards &  36s &    2365MB \\
 Large number &   53s &   1218MB \\
 $M\leq 3$ Muddy children & 3s & 130Mo \\
 $M\leq 4$ Muddy children & 3s & 162Mo \\
 $M\leq 9$ Muddy children & 24s & 1136Mo \\
 $M\leq 10$ Muddy children & TO~(5min+) &  \\
 $M\leq 11$ Muddy children & out of memory & \\
 $M<\infty$ Muddy children & 2.5s  &  111MB \\
\end{tabular}\\
The rest of the section discusses implementation details and description of the
aforementioned models.

\textbf{Usage.}
The tool takes as an input an automaton description of a regular Kripke
structure~$\calM$, and for each specification~$\varphi$,
computes its satisfaction set.
In case the complement~$\sem{\neg \varphi}(\calM)$ is non-empty,
a NFA is returned, which can be interpreted as the
set of counterexamples to~$\varphi$.
For usability reasons, the syntax of \PPALstar{} is
enhanced with several syntactic sugars, but can also embed dummy formulae,
equivalent to $\top$, whose
evaluation triggers visualization of the intermediate constructed automata.

\textbf{Automaton size.}
Since specifying a transducer for $\obs{\cdot}$ can be quite
tedious, we specify a rather general regular
Kripke structure encoding only the observation of the agents,
and further restricting the state space by applying
public announcement constructions.
As a matter of fact, the state space after only few announcements can already
require several hundred states.
The intermediate computations
may even lead to semantics automata of up to millions of states.
Note that the ordering of index quantifications inside the specification
plays a crucial role, as each quantified index is carried around in one coordinate
of the automaton alphabet, as explained by \cref{lem:recursive}.

\textbf{Learning procedure.}
Although several
DFA learning algorithms are provided by Learnlib, the classical
Angluin's~\Lstar turned out to be sufficient for our experiments:
for all our examples,
whenever termination was guaranteed\footnote{The learning procedure
diverges if, and only if, $L_\preceq$ is not regular.},
the algorithm converged within a minute.
The most expensive
task of the equivalence check is the last property of
\cref{prop:implemchr}:
it is indeed the only criterion involving the evaluation of the \PPAL{}
formula. Fortunately, many equivalence queries fail on previous criteria,
that are less expensive to check.

\subsection{Russian Card Problem}
\label{sub:russian}

This puzzle~\cite{vD03} involves $N$ different cards which are distributed
between three players Alice, Bob and Cathy. The goal of the game is for Alice
and Bob to exchange messages publicly, in order to get to know who has which
card in their hand, without disclosing any individual card information to Cathy.

In the one-round setting, Alice broadcasts a first
message, then Bob replies, which conclude the protocol.
As Bob can only announce a piece of information he already knows,
his message can trivially be assumed to announce Cathy's cards.
In other words, the one-round case focuses on Alice's announcement.

\textbf{Kripke structure.}
We let
$\AP=\{a,b,c\}$, and the only agent indexes
involved\footnote{We can assume that $\obs{i}$ is trivial for $i\geq 3$.}
are $a=1$, $b=2$, and $c=3$.
For $x\in \AP$ and $i\in \bbN$, $x_i$ holds iff agent $x$ has card $i$ in
their hand. Moreover, we assume that $a_i,b_i$ and $c_i$ are mutually exclusive
(each card appears only in one) hand. We easily check that $\calM$ is regular.

\textbf{Specification.}
An announcement of Alice is any statement about her own observation,
namely a characterization of the cards in her hand, or equivalently
$\left\{\{i_0, i_1, i_2\},
\{i_3, i_4, i_5\}, \hdots\right\}$,
seen as a set of
possible hands.
However, this representation is not fit to a parameterized context, where the
total number of cards is not fixed a priori.
Instead, we consider announcements specified in a
parameterized manner, namely in the propositional
fragment of \PPAL{}, involving only index quantifications,
the atomic proposition~$a$ and no epistemic operator.
A formula $\psi$ is a \emph{good announcement} if furthermore, it satisfies:
\[
\begin{aligned}
  \varphi_{{good}} = \quad
 & \psi \wedge \ann{ \psi }{(}
    & \text{// truthful PA}\\
 &\forall i,
    \knows{b}a_i \vee
    \knows{b}b_i \vee
    \knows{b}c_i
    & \text{// b knows the distribution}\\
 &\forall i, \neg c_i \rightarrow
    \left\{\begin{aligned}
    & \may{c}a_i \wedge 
    \may{c}\neg a_i \\ 
    & \may{c}b_i \wedge 
    \may{c}\neg b_i \quad )\hspace{-10em}
   \end{aligned}
   \right.
   & \text{// c doesn't know}
\end{aligned}
\]


While~\cite{vD05} provides several sufficient and necessary conditions on the number
of cards received by each participants, we focus here on a single example
of (sufficient) good announcement, provided by~\cite[Proposition~5]{vD05} in the
case where $N\%3 = 0$, Alice receives $3$ cards, Cathy only one, and Bob the rest
(property $\varphi_{{model}}$).

If Alice received the first three cards, the following announcement is claimed
to be good:
\[
    \psi \equiv \exists j: j\%3 = 0 \wedge
      \left(
        (a_{j} \wedge a_{j+1} \wedge a_{j+2})
      \vee
        (a_{j} \wedge a_{j+4} \wedge a_{j+8})
      \right)
\]
Which can be checked with the verification question:\[
    \calM \stackrel{?}{\vDash}
        \ann{\varphi_{{model}}}{
          \left(
            a_0 \wedge a_1 \wedge a_2 \rightarrow \varphi_{{good}}
          \right)
        }
\]

{
We leave to the reader the generalization to any initial hand of Alice and the
specification of $\varphi_{{model}}$.
}{
\[
  \varphi_{{model}} \equiv
    \left\{
    \begin{aligned}
      &  N\%3 = 0 & \\
      &  \exists i: c_i 
      \wedge \forall j:~ c_j \rightarrow i = j \\
      &  \exists i,j,k: i\neq j \wedge i\neq k \wedge j\neq k \wedge a_i \wedge a_j \wedge a_k \\
      &  \forall l, a_l \rightarrow
            i = l \vee j = l \vee k = l
    \end{aligned}\right.
\]

\begin{remark}
If Alice is not given the cards $0,1,2$, nor another combination specified
by~$\psi$, the announcement is not valid. However, it can be rewritten, depending
on Alice's current hand.
Let $x$ and $y$ be two index variables. For any $\varphi$ formula, we denote
$\varphi[x|y]$ for the formula where any propositional sub-formula $p_i$ has
been rewritten into
$i = x \wedge a_y \vee i = y \wedge a_x \vee i \notin \{x,y\} \wedge a_i$.

The previous verification question is converted into:
\[
    \calM \stackrel{?}{\vDash}
        \ann{\varphi_{{model}}}{
          \left(
            \exists x_1,y_1,x_2,y_2: \knows{a}
            \varphi_{{good}}(\psi[x_1|y_1][x_2|y_2])
          \right)
        }
\]
Note that the choice of a satisfying set of indices
$x_1, x_2,y_1,y_2$
must not be serendipity: it should work for all possible hands of $b$ and $c$,
that $a$ may imagine, hence the universal $\knows{a}$ quantification.
Note also that we need to swap only two pairs of cards, to reconstitute
a triple of cards appearing in~$\psi$.
\end{remark}
}

\subsection{Highest number}
\label{sub:highest}
The highest number problem involves two agents Alice and Bob both receiving
a different natural number between $0$ and $N$,
which they keep private. We model this situation
by $\AP=\{a,b\}$ and encode the observation of $a=0$ and $b=1$ as a transducer.
A letter $\alpha\in \AP$ is used to encode $\alpha$'s number,
in unary.
At each round, they are both asked simultaneously if one of them knows who has
the highest number. If not, a public announcement is made for this fact. If yes, the game
stops.
The termination of this protocol is checked by the
following iterated announcement, which we successfully verify:
$
    \annstar{
        \neg\left(
           \exists i\exists j: \knows{j} (a_i \wedge \neg b_i) \vee
                       \knows{j} (\neg a_i \wedge b_i)
        \right)
    }{\bot}{*}
    $
\subsection{Muddy Children: Bounded Case}
\label{sub:bounded}
In this section, assume the number of muddy
children is bounded by some fixed $M\in\bbN$,
although the total number of children is left as a parameter~$N$.
This assumption is implemented as public announecement made on the regular
Kripke structure of \cref{ex:muddyreg}.
{
   \[
      \ann{\exists i: m_i}{} 
      \ann{\exists i_1 \hdots i_M: \forall j: m_j \rightarrow \bigvee_{k=1}^M j=i_k}{}
    \]
}

Intuitively, the effect of this announcements is to
construct the product automaton of the original transducer~$\trans{\calM}$ with
a finite automaton of size $M+1$ counting how many muddy children have been
seen so far. This product has to be made twice: once on the source and once
on the target of the transducer. Nonetheless,
the target and source word differ only by one letter, hence the resulting
automaton is of size $O(M)$.

Then, we proceed to the iterated announcement
$\annstar{\forall i, \may{i}\neg m_i}{}{*}$, which reduces to the disappearance
relation computation: for $s,t\in \St_0$, $s\preceq t$ if, and only if,
$|s|_m \leq |t|_m$. As a matter of fact, the protocol terminates
after $|s|_m$ announcements of the father
whenever there are exactly $|s|_m$ muddy children.

This relation can be effectively encoded as a length-preserving
transducer, counting the difference of number of muddy children between $s$ and
$t$, which
lies between $-M$ and $M$. Our algorithm successfully computes
a transducer for $\preceq$, with $O(M)$ states.

\subsection{Unbounded Case and Symmetry Reduction}
\label{sub:muddy}
We remove now the boundedness condition. As before, $\preceq$ has to
compare the number of muddy children between two given states, which can now
be arbitrarily large: take for example $m^nc^n \preceq c^nm^n$.
As a consequence, $L_\preceq$ is not regular anymore
and the learning procedure doesn't terminate.

Nonetheless, we observe that the problem is invariant under permutation,
more precisely:
  (1) The formula $\varphi$ lies in a fragment of \PPALstar{} without
    index comparison of the form $i=j+k$ for any $k\neq 0$;
  (2) For
    any word $w\in\trans{\calM}$ and any bijection $\Sigma$ on $\range{|w|}$,\\
    $w[\sigma(0)] \hdots w[\sigma(|w|-1)]\in \trans{\calM}$.

Therefore, we proceed to a counting abstraction of the model, restricting
the regular Kripke structure. Informally,
we want to preserve the property that
a transition $s \obs{i} t$ is valid if, and only if,
there exists some agent~$j$ with the same "local state" as~$i$, that can perform
this transition. Here, the announcement translates to \enquote{there is still a
muddy child who doesn't know}.

As the state space is reduced to $c^*m^*$, our rewriting actually consists
in a unary encoding of the number clean and muddy children.
As for the largest number
challenge, the disappearance relation is regular, and we successfully
verify the rewritten formula:
\[
\varphi\equiv
    \ann{\forall i, \neg m_{i+1} \rightarrow \neg m_i}{}
    \ann{\exists i: m_i}{}
    \annstar{\exists i: m_i \wedge \may{i}\neg m_i}{\bot}{*}
\]

\section{Related and future work}
\label{sec:conclusion}

\noindent
\textbf{Related Work.}
Finite-state model checkers for various epistemic logics are available, e.g., 
MCMAS~\cite{LQR17}, DEMO~\cite{DEMO,DR07}, SMCDEL~\cite{benthem}, and MCK~
\cite{MCK}. Kouvaros and Lomuscio~\cite{KL16} have studied \emph{cutoff 
techniques}
for $\text{ACTL*K} \setminus \text{X}$, a temporal-epistemic logic combining
S5 and temporal logic $\text{ACTL*} \setminus \text{X}$, which is used in
MCMAS. Roughly speaking, a cutoff exists for a parameterized system when the 
behavior of any instance of the system
can be simulated (using an appropriate notion of simulation) by the behavior of 
systems of a fixed computable parameter-size $k$, which would allow us to 
reduce the
parameterized model checking problem into finite-state model checking (up to
parameter of size $k$). This cutoff technique --- as is the case with most 
cutoff methods (see~\cite{vojnar-habilitation,sasha-book}) --- needs to be 
specially tuned to different subclasses of parameterized systems. We are not
aware of the existence of such cutoff values for the systems that we consider
in this paper.
Our regular model checking method is complementary to such cutoff methods.
The method is fully automatic, but it might not terminate in general (albeit we
provide also termination guarantees). To the best of our knowledge, our
method provides the first automatic solution to the 
parameterized verification problem for the muddy children puzzle, the Russian
card puzzle \cite{vD03}, and the large number challenge, all of which have been
studied in the finite-state case (e.g. see~\cite{benthem,LQR17,DR07,DEMO}).

\noindent
\textbf{Future Work.}
Natural extensions of \PPALstar{} include the support
of dynamic properties, enabling the specification and verification
of richer communication protocols, where the communication pattern is
non-deterministic~\cite{LQR17}.
The study of the disappearance relation revealed that the chosen encoding is
crucial for termination. The counting abstraction sketched for the Muddy children
case would benefit from a systemic approach. Once the symmetries have been
detected in the automatic structure, which can be implemented~\cite{vmcai16}
with transducer techniques, a lossless Parikh image~\cite{P66} could be
computed in terms of a Presburger formula~\cite{SSMH04}. As the \PPAL{}
semantics involves only boolean, \crossp{} and morphism operations, the computation
could be performed in this domain. As another future direction, we would also
like to study cutoff methods of~\cite{KL16}
for the examples that we consider in this paper.

\begin{acks}
This work was 
  supported by the ERC Starting Grant 759969 (AV-SMP) and Max-Planck Fellowship.

\end{acks}



\bibliographystyle{ACM-Reference-Format} 
\bibliography{main}


\end{document}